\documentclass[runningheads]{llncs}
\usepackage{graphicx}
\usepackage{amsmath, amssymb}
\usepackage{graphicx, subfigure}
\usepackage[ruled,linesnumbered]{algorithm2e}
\usepackage{makecell}
\usepackage{xcolor}
\usepackage{float}
\usepackage{tcolorbox}
\usepackage{url}

\usepackage{enumitem}
\setlist{nolistsep,leftmargin=*}
\usepackage[skip=1pt]{caption}

\def\sB{\boldsymbol B}
\def\sC{\boldsymbol C}
\def\sH{\boldsymbol H}

\def\SPAR{\textbf{SPAR}}

\def\S{\mathcal S}

\begin{document}
\title{Coded Merkle Tree: Solving Data Availability Attacks in Blockchains}

\author{Mingchao Yu\inst{1} \and 
              Saeid Sahraei\inst{1} \and 
              Songze Li\inst{2} \and
              Salman Avestimehr\inst{1} \and\\ 
              Sreeram Kannan\inst{2,3} \and 
              Pramod Viswanath\inst{2,4}}

\authorrunning{M. Yu et al.}

\institute{University of Southern California \\
\email{fishermanymc@gmail.com}, \email{ss\_805@usc.edu}, \email{avestimehr@ee.usc.edu} \and
Trifecta Blockchain\\
\email{songzeli8824@gmail.com}\and
University of Washington Seattle\\
\email{ksreeram@uw.edu} \and 
University of Illinois at Urbana-Champaign \\
\email{pramodv@illinois.edu}}

\maketitle
\begin{abstract}
In this paper, we propose coded Merkle tree (CMT), a novel hash accumulator that offers a constant-cost protection against data availability attacks in blockchains, even if the majority of the network nodes are malicious. A CMT is constructed using a family of sparse erasure codes on each layer, and is recovered by iteratively applying a peeling-decoding technique that enables a compact proof for data availability attack on any layer. Our algorithm enables any node to verify the full availability of any data block generated by the system by just downloading a $\Theta(1)$ byte block hash commitment and randomly sampling $\Theta(\log b)$ bytes, where $b$ is the size of the data block. With the help of only one connected honest node in the system, our method also allows any node to verify any tampering of the coded Merkle tree by just downloading $\Theta(\log b)$ bytes. We provide a modular library for CMT in {\sf Rust} and {\sf Python} and demonstrate its efficacy inside the {\sf Parity Bitcoin} client. 
\end{abstract}

\vspace{-10mm}
\section{Introduction}
\vspace{-3mm}
Blockchains (e.g., Bitcoin~\cite{nakamoto2008bitcoin} and Ethereum~\cite{wood2014ethereum}) maintain a ledger of ordered transactions,  organized into a chain of blocks. Starting from the genesis block, network nodes extend the ledger by creating and appending more blocks, following specific block generation rules (e.g., the longest-chain rule is used in Bitcoin \cite{nakamoto2008bitcoin}). The transactions in the received blocks are validated by {\em full nodes} which  download the entire block tree.
However, for better scalability, it is imperative for a blockchain to allow {\em light nodes}, which may only be interested in verifying some specific transactions.   

In Bitcoin~\cite{nakamoto2008bitcoin,bitcoinOperatingModes}, light nodes are implemented using the Simple Payment Verification (SPV) technique:  a Merkle tree is constructed for each block using the transactions as the leaf nodes, and the Merkle root is stored in the block header. Utilizing the Merkle root, a light node can verify the inclusion of any transaction in a block through  a Merkle proof. Light nodes and SPV have been leveraged extensively to scale computation and storage of blockchain systems over resource-limited nodes (e.g., smartphones)~\cite{cryptonite,BIP37,Electrum,LES,mcconaghy2016bigchaindb,xu2017epbc,frey2016bringing,dorri2017lsb,gervais2014privacy}.

Besides inclusion, what is more important for a light node is to validate the transaction based on the ledger state. Due to limited resources, a light node cannot download the entire ledger. Instead, it could use the depth of the block that contains this transaction as a proxy. That is, the deeper this block is buried into the chain, the more confident the light node is about the validity of the transaction. However, for it to work, a majority (in terms of hashing power, stakes, etc.) of full nodes must be honest and must follow protocol. Further, there is a significant \emph{tradeoff} between confirmation latency (due to the depth) and the security about transaction validity.



Therefore, efforts to study 1) the scenario of a light node being connected to {\em dishonest majority of full nodes}, and 2) how to achieve faster confirmation at light nodes are becoming a major research direction \cite{nakamoto2008bitcoin,mustafa_fraud_proof,eyal2016bitcoin,bano2017road}. The overall idea is to design new block structures that allow full nodes to generate and broadcast {\em succinct fraud proofs} of individual transactions. This way, a light node will be able to timely verify fraud transactions and blocks as long as it is connect to \emph{at least one} honest full node.
One efficient construction that utilizes the roots of the intermediate state Merkle trees after executing a subset of transactions 
is proposed in~\cite{mustafa_fraud_proof}. However, it is vulnerable to the so-called ``data availability attack'' described in \cite{mustafa_fraud_proof}, for which \cite{mustafa_fraud_proof} proposed an erasure code based solution. Stating the data availability attack formally and solving it comprehensively is the main goal of this paper.



\noindent{\bf Data availability attack}. A malicious block producer 1) publishes a block header, so that light nodes can check transaction inclusion; but 2) withholds a portion of the block (e.g., invalid transactions), so that it is impossible for honest full nodes to validate the block and generate the fraud proof.

Although the honest full nodes are aware of the data unavailability,  there is no good way to prove it. The best they can do is to raise an alarm without a proof. However, this is problematic because the malicious block producer can release the hidden parts \emph{after} hearing the alarm. Due to network latency, other nodes may receive the missing part before receiving the alarm and, thus, cannot distinguish who is prevaricating. Due to this, there is no reward and punishment mechanism that can properly reward honest full nodes while also deterring false alarms and denial-of-service attacks.


Therefore, for fraud proofs to work, light nodes must determine data availability by themselves. This leads to the following key question: {\em when a light node receives the header of some block, how can it verify that the content of that block is available to the network by downloading the least possible portion of the block}?


\noindent {\bf Need to encode the block}. Since a transaction is much smaller than a block, a malicious block producer only needs to hide a very small portion of a block. Such hiding can hardly be detected by light nodes unless the entire block is downloaded. However, by adding redundancy to the data through appropriate erasure codes \cite{lin2001error}, any small hiding on the origin block will be equivalent to making a significant portion of the coded block unavailable, which can be detected by light nodes through randomly sampling the coded block with exponentially increasing probability. As a counter measure, a malicious block producer could instead conduct coding incorrectly to prevent correct decoding. Light nodes rely on honest full nodes to detect such attacks and prove it through an incorrect-coding proof.

For example, an $(n,k)$ Reed-Solomon (1D-RS) code \cite{reed1960polynomial} encodes $k$ data symbols into $n$ coded symbols, and any $k$ out of these $n$ coded symbols can be used to decode the $k$ data symbols. Thus, to prevent decoding, a malicious block producer will have to make at least $n-k+1$ coded symbols unavailable, which yields a detection probability of $1-(1-k/n)^s$ after sampling $s$ distinct coded symbols uniformly at random. But an incorrect-coding proof will consist of $k$ coded symbols, which is of the same size as the original block and thus is too large. This cost is alleviated to $\sqrt{k}$ in \cite{mustafa_fraud_proof} by using two-dimensional RS codes (2D-RS), at the costs of reduced sampling efficiency, and increased block hash commitments of $2\sqrt{n}$ Merkle roots to verify the coding correctness within each dimension. In addition, 1D-RS and 2D-RS codes have a high decoding complexity of $O(k^2)$ and $O(k^{1.5})$, respectively.

In summary, with erasure coding, a light node pays 3 download costs for data availability, including block hash commitments, symbol sampling, and incorrect-coding proofs. Among them, the incorrect-coding proof cost must be minimized to defend fake proofs, for which both 1D-RS and 2D-RS are sub-optimal.

\noindent {\bf Our contributions}.
In this paper, we propose SPAR (SParse frAud pRotection),
the first data availability solution that promises order-optimal performance on \emph{all} the metrics, including 1) a constant block hash commitment size; 2) a constant sampling cost for a given confidence level on data availability; 3) a constant incorrectly-coding proof size; and 4) linear decoding complexity (Table \ref{tab:light_node_costs}).

\vspace{-3em}
\begin{table}[h]
    \centering
\caption{Light node download costs and full node decoding complexity ($b$: block size in bytes).}
    \scalebox{0.87}{
    \begin{tabular}{|c|c|c|c|c|}
    \hline
        ~&\makecell{hash commitment\\size (bytes)} & \makecell{\# of samples to gain\\ certain confidence\\ about data availability}&\makecell{incorrect-coding\\ proof size (bytes)}&\makecell{decoding\\complexity}\\\hline
        Uncoded & $O(1)$ & $O(b)$ & - & -\\\hline
         1D-RS& $O(1)$ & $O(1)$ & $O(b\log b)$  & $O(b^2)$ \\\hline
         2D-RS \cite{mustafa_fraud_proof}& $O(\sqrt{b})$ & $O(1)$ & $O(\sqrt{b}\log k)$&$O(b^{1.5})$\\\hline
         \SPAR& $O(1)$ & $O(1)$ & $O(\log b)$&$O(b)$\\\hline
    \end{tabular}
    }
    \label{tab:light_node_costs}
\end{table}
\vspace{-1.5em}

At the core of SPAR is a novel cryptographic hash accumulator called coded Merkle tree (CMT). Starting from the bottom, CMT iteratively encodes layers of the tree and uses the hashes of the coded layer as the data for the next layer. 
A light node can detect the availability of the \emph{entire tree} through the Merkle proofs of bottom layer leaves. With the entire tree available, SPAR uses a novel hash-aware peeling decoder and a special ensemble of random LDPC (low-density parity-check) codes to maximize sampling efficiency, minimize incorrect-coding proof to one parity equation, and achieves linear decoding complexity.

\noindent {\bf  SPAR and CMT implementation}. We have developed a complete and modular CMT library in {\sf Rust} and {\sf Python} \cite{cmt-lib}. We have also implemented SPAR in the {\tt Bitcoin Parity} client \cite{parity-btc}, which outperforms state of the art \cite{mustafa_fraud_proof} by more than 10-fold in hash commitments, incorrect coding proof, and decoding speed.     


\noindent {\bf Related works.} This work was inspired by pioneering research in \cite{mustafa_fraud_proof}, which proposes succinct fraud proofs and 2D-RS based data availability solution. Besides this work, coding also improves  scalability of blockchains in other areas: \cite{perard2018erasure} studies the coding efficiency of distributed storage systems~\cite{aguilera2005using,dimakis2011survey,dimakis2010network,rashmi2011optimal}. 
In a related vein, \cite{raman2017dynamic} uses a combination of Shamir's secret sharing \cite{shamir1979share} (for storing the headers and the encryption keys) and private distributed storage (for the blocks) to reduce the storage overhead while guaranteeing data integrity and confidentiality. \cite{li2018polyshard} uses Lagrange coding to simultaneously scale storage, computation, and security in a sharded blockchain~\cite{luu2016secure,kokoris2018omniledger}, via cross-shard coding. 


\vspace{-5mm}
\section{Security Model}
\vspace{-3mm}
The core functionality of a blockchain  is to produce,  verify, and accept/store valid data blocks in a consistent but decentralized manner. A data block, denoted by $\sB$, is a byte string of length $b$ that carries a batch of transactions. $\sB$ is valid for acceptance if and only if {\em every} single transaction in it is valid (e.g., enough balance, no double spending). 
Thus {\em incomplete} data blocks are tantamount to being unacceptable.
Data incompleteness is not a threat to a node that fully downloads the block. However, state-of-the-art blockchain systems also run light nodes which do not download the  blocks in entirety. We next describe these two types of nodes formally (see Fig.~\ref{fig:model}). 

\emph{Full nodes} are able to produce blocks (e.g., by batching submitted transactions), and to download and verify blocks produced by other full nodes. Upon acceptance, they store the entire block locally. Upon rejection, they broadcast a fraud proof to alert the other nodes. We note, however, that malicious full nodes do not necessarily follow such requirements, and can act arbitrarily. 

\emph{Light nodes} can only afford to download a small amount of data from each block and perform simple computations such as hash checks and fraud proof verification, but not to operate on whole data blocks. By accepting a block $\sB$, they only store its hash commitment $D=g(\sB)$. Here $g()$ is a hash accumulator such as Merkle tree generator, which will allow it to use $D$ to verify the inclusion of any transaction in $\sB$ through a Merkle proof. Without loss of generality, we assume light nodes are honest, as they are not able to deceive full nodes.


\begin{figure*}[t]
    \centering
 \includegraphics[width=0.4\textwidth]{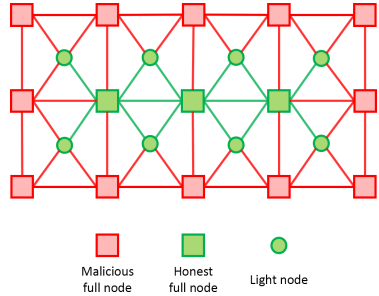}
 \caption{The connections between full nodes and light nodes. Every honest (green) node is connected to at least one honest full (square) node. A light (circle) node may connect to more malicious (red) full nodes than honest full nodes.}
    \label{fig:model}
    \vspace{-6mm}
\end{figure*}

We assume the following network model:

1. Reliable communication: Two directly connected nodes can reliably communicate via both unicast and broadcast without message loss or corruption.

2. Connectivity: Every honest (full and light) node is directly connected to at least one honest full node. In other words, we assume a connected sub-graph of honest full nodes, and all the light nodes are connected to it.
 
3. The network is synchronous. These three assumptions together means that a valid message sent by an honest node can be received by all appropriate honest nodes (e.g. blocks for honest full nodes, block headers and fraud proofs for light node) within a fixed delay if every honest node re-broadcasts it.

4. The network allows nodes to send messages anonymously.


Importantly, our network model allows  {\emph{dishonest majority}}, i.e.,  each light node can be directly connected to more malicious full nodes than honest ones. Due to this, a light node cannot determine the completeness of a block through its connected full nodes,  via a majority vote for instance.


A malicious block producer is thus motivated to conduct a {\em data availability attack}, where it 1) does not fully disclose $\sB$, so that honest full nodes are not able to verify $\sB$; and 2)
broadcasts $D$, so that itself and its colluding full nodes can forward $D$ to their connected light nodes and deceive them that the $\sB$ that satisfies $g(\sB)=D$ is valid for accepting. Thus, the key for a light node to protect itself from accepting a fraudulent block is to make sure that $\sB$ is fully available. This gives rise to the main problem we try to address in this paper:

\begin{tcolorbox}
\noindent \textbf{Data availability problem:} \emph{Upon receiving a hash commitment $D$, how can a light node efficiently verify that a data block $\sB$ that satisfies $g(\sB)=D$ is fully available to the system?}
\end{tcolorbox}

A simple strategy for a light node is to randomly sample portions of $\sB$, and determine that it is unavailable if it does not receive all requested portions. Since the size of a transaction is usually much smaller than the block, a malicious block producer only needs to hide a very small portion (say, e.g., a few hundred bytes) of a fraudulent block, which can hardly be detected through random sampling.

A malicious block producer could also conduct {\em selective disclosure}: when requested by light nodes, the malicious block producer may select a subset of the light nodes and fully disclose their requested portions, as long as the total disclosed portions do not reveal $\sB$. These light nodes will be deceived about the availability of $\sB$ and will accept it, as no fraud proof of $\sB$ can be produced.

\noindent Thus, as similarly done in~\cite{mustafa_fraud_proof}, we characterize the security of the above described system using the following measures: 

\textbf{\emph{Soundness}}: If a light node has determined that a data block is fully available, then at least one honest full node will be able to fully recover this data block within a constant delay.

\textbf{\emph{Agreement}}: If a light node has determined that a data block is fully available, then all the other light nodes in the system will determine that the data block is fully available within a constant delay.

Recently, an erasure coding-assisted approach was proposed in \cite{mustafa_fraud_proof} to improve sampling efficiency and suppress the data availability attack. In the next section, we will motivate this approach and overview the challenges it faces.

\vspace{-3mm}
\section{Overview of Erasure Coding Assisted Approach}
\vspace{-3mm}
An $(n,k)$ erasure code evenly partitions a block $\sB$ of $b$ bytes into $k$ data symbols of $\frac{b}{k}$ bytes each as $\sB=[m_1,\cdots,m_k]$, and linearly combines them to generate a coded block with $n>k$ coded symbols, $\sC=[c_1,\cdots,c_n]$. The ratio $r=k/n$ is called the coding rate. The $n$ hashes of these coded symbols are accumulated to obtain the hash commitment $D$ of $\sC$, which is published with $\sC$. With a good erasure code, a block producer's hiding of one data symbol is equivalent to making the value of many coded symbols unavailable to the system. In general, a pair of good erasure code and decoding algorithm yields a large \textbf{\em undecodable ratio} $\alpha$, which is the minimum fraction of coded symbols a malicious block producer needs to make unavailable to prevent full decoding. Such hiding can be caught by a light node with an exponentially increasing probability of $1-(1-\alpha)^s$ through randomly sample $s$ coded symbols when $n$ is large, indicating that $O(1)$ samples are sufficient. Below is an example.

\begin{example}\label{example:uncoded_coded}
\emph{\textbf{Uncoded v.s. coded sampling efficiency.} Given a block of 4 data symbols $[m_0, \cdots, m_3]$, a block producer generates 8 coded symbols as follows:}
\begin{equation}\label{eq:encoding_example}
\begin{cases}
c_0=m_0, ~c_1=m_1, ~c_2=m_2, ~c_3=m_3,\\
c_4=c_0+c_1,~c_5=c_1+c_2,~c_6=c_2+c_3,~c_7=c_3+c_0.
\end{cases}
\end{equation}
\noindent
\emph{To prevent decoding through hiding, a malicious block producer must either publish no more than 3 data symbols or no more than 5 coded symbols. Both will make at least 3 coded symbols unavailable to the system ($\alpha=\frac{3}{8}$). Such unavailability can be caught with a probability of $1-\frac{5}{8}\cdot\frac{4}{7}=64.3\%$ after randomly sampling 2 distinct coded symbols. In contrast, without coding, the hiding of one data symbol can be caught with a probability of only  $1-\frac{3}{4}\cdot\frac{2}{3}=50\%$.}
\end{example}

To counter erasure coding assisted random sampling, a malicious block producer could conduct an \textbf{\em incorrect-coding attack}: It generates coded symbols that fail the parity equations (the equations describing the linear relations between coded symbols in Example \ref{example:uncoded_coded}) specified by the erasure code, and generates $D$ using these invalid coded symbols. This way, it can pass light node random sampling through hiding only one data symbol and publishing most of the coded symbols, which will not allow honest full nodes to correctly decode $\sB$.

Fortunately, this attack can be detected by honest full nodes by comparing the decoded block with the commitment $D$. Upon detection, an honest full node can generate an \textbf{\em incorrect-coding proof}, which consists of the coded symbols of failed parity equation(s) and appropriate hash commitments, so that light nodes can verify them and reject the block. Using Example \ref{example:uncoded_coded}, an incorrect coding proof about $c_4=c_0+c_1$ could be 
$c_0$ and $c_1$ with matching Merkle proofs, plus the Merkle proof of $c_4$, which, however, does not match the value of $c_0+c_1$.

To keep incorrect coding proofs small, \cite{mustafa_fraud_proof} applies 2D-RS (2-dimensional Reed-Solomon) code. The $k$ data symbols are placed as a $\sqrt{k}\times\sqrt{k}$ square, then a $(\sqrt{n}, \sqrt{k})$ RS code is applied to every row/column. The resulted $2\sqrt{n}$ rows/columns yield $2\sqrt{n}$ Merkle roots, which are downloaded by light nodes as block header. Each root allows a light node to verify the associated row/column by decoding it using any $\sqrt{k}$ coded symbols of it (from incorrect-coding proof) and reproducing the root. Thus, 2D-RS offers light nodes 1) a header cost of $O(\sqrt{b})$, 2) a sampling cost of $O(\log b)$, and 3) an incorrect-coding proof size of $O(\sqrt{b}\log b)$. Here $\log b$ is due to logarithmic growth of Merkle proof size with $b$.

In this paper, we propose SPAR (SParse frAud pRotection), the first solution to the data-availability problem that is order-optimal in all the above three metrics: a header cost of $O(1)$, a sampling cost of $O(\log b)$, and an incorrect-coding proof size of $O(\log b)$. To this end, SPAR leverages four core components:
\begin{enumerate}
    \item a novel hash accumulator named coded Merkle tree (CMT), which encodes every layer of the tree to protect the availability of the entire tree. This way, the Merkle proof of every coded symbol will be available, which will enable every parity equation to be committed and verified \emph{alone};
    \item a dedicated sampling mechanism that enables a light node to check the availability of the entire CMT by sampling $O(\log b)$ bytes plus one Merkle root;
    \item a hash-aware decoding algorithm that is able to detect and prove any single failed parity equation, provided the Merkle proofs of all the coded symbols;
    \item a special ensemble of random LDPC (low-density parity check) codes with a constant parity equation size and a constant undecodable ratio under the above hash-aware decoding algorithm, which protects all CMT layers equally.
\end{enumerate}
\vspace{-4mm}
\section{Detailed Description of SPAR}\label{sec:SPAR_details}
\vspace{-3mm}
In this section, we describe the four core components of SPAR: the construction of the coded Merkle tree by the (honest) block producer, the sampling mechanism of the light nodes, the decoding and alerting operations of the honest full nodes, and the erasure codes used by SPAR. At the end of this section, we will summarize the action space of each node in the network.

\vspace{-5mm}

\subsection{Construction of coded Merkle tree}
\label{sec:cmt_construction}
\vspace{-1mm}
In SPAR, an honest full node detects and proves incorrect-coding
using the membership proofs of \emph{all} the $d$ coded symbols in one parity equation and the values of at least $d-1$ of these coded symbols. Since any parity equation can be compromised, a light node needs to make sure the membership proofs of \emph{all} the $n$ coded symbols are available at honest full nodes. In other words, it needs to make sure the entire Merkle tree is available.

To this end, we propose CMT. At a high level, CMT applies erasure coding to every layer of the tree, where the data symbols of a layer are generated using the hashes of the coded symbols of its child layer. This way, a light node can check the availability of every layer through random samplings, whilst an honest full node can detect and prove the incorrect coding at any layer, with the help of the hashes of this layer provided at its parent layer.

More specifically, given a block of $k$ data symbols, a rate-$r$ ($r \leq 1$) systematic erasure code with an undecodable ratio of $\alpha$ is applied to generate $n=k/r$ coded symbols, where the first $k$ are the original data symbols and the remaining $n-k$ are called parity symbols (hence the name systematic). Then the hashes of every $q$ coded symbols are batched as one data symbol for the next (parent) layer. This yields a total of $n/q$ data symbols for the next layer, which will be encoded using a smaller (in terms of $k$) rate-$r$ systematic code with the same undecodable ratio. This iterative encoding and batching process stops once there are only $t$ ($t\geqslant 1$) hashes in a layer. These $t$ hashes are the root of the CMT, and will be included in the block header and published with the original data block.

CMT layer size reduces at a rate of $qr$. Thus, $qr>1$ for CMT to converge.  In addition, to enable efficient sampling of both data and parity symbols (will discuss next), batching is interleaved, namely, the $q$ coded symbols whose hashes are batched together consist of $qr$ data symbols and $q(1-r)$ parity symbols. An example of CMT with $k=16$, $r=\frac{1}{2}$, $q=4$, and $t=4$ is illustrated in Fig.~\ref{fig:coded_Merkle_tree_interleaved}. Indeed, a classic Merkle tree is a special CMT with $r=1$ and $q=2$.

\begin{figure*}[t]
    \centering
    \includegraphics[width=\linewidth]{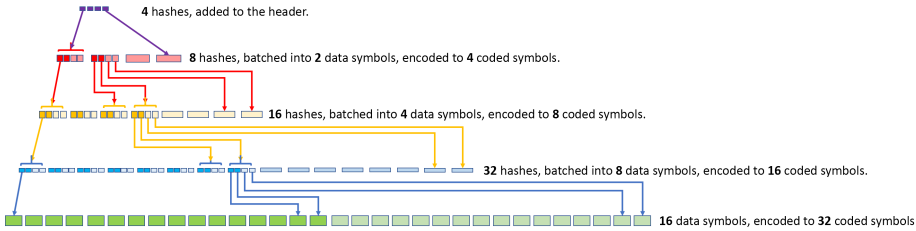}
    \caption{Coded Merkle tree for $k=16$ data symbols using rate $r=\frac{1}{2}$ erasure codes. Each data symbol of a higher layer is constructed by batching the hashes of $qr=2$ data symbols and $q(1-r)=2$ parity symbols of its child layer.}
    \label{fig:coded_Merkle_tree_interleaved}
    \vspace{-6mm}
\end{figure*}

\vspace{-5mm}

\subsection{Sampling Mechanism of Light Nodes}
\label{sec:samplingalgorithm}

\vspace{-2mm}
In SPAR, a light node randomly samples the base layer coded symbols with their Merkle proofs to decide the availability of the base layer. The special structure of CMT allows them to further utilize these proofs to efficiently sample higher layer symbols to decide the availability of higher layers.

Similar to a classic Merkle tree, the Merkle proof of a base layer symbol in CMT consists of all the sibling hashes between this symbol and the root. The only difference is that the number of sibling hashes per layer is now $q-1$ instead of 1, which effectively provides the light node one data symbol from every intermediate layer. Thus, when a light node randomly samples $s$ distinct base layer coded symbols, the associated Merkle proofs will automatically sample,  at \emph{no} extra cost, $s$ distinct data symbols from every intermediate layer w.h.p.

To properly check the availability of an intermediate layer,
a light node should also randomly sample about $(1-r)s$ parity symbols from this layer. To avoid downloading extra Merkle proofs for these parity symbols and to minimize the correlation between the samplings\footnote{Otherwise, the malicious block producer can hide the highly correlated symbols of the same layer together to reduce light nodes' detection probability.}, SPAR samples parity symbols of intermediate layers probabilistically: For every pair of parent and child intermediate layer, if a parent layer data symbol is sampled, then with probability $1-r$, one of its $q(1-r)$ child parity symbols (thanks to interleaved batching) will be sampled uniformly at random. 
Thus, the response size of one sampling request will be:
\vspace{-0.5em}
\begin{equation}\label{eq:merkle_proof_with_parity}
    \frac{b}{k}+\left[y(q-1) + yq(1-r)\right]\log_{qr}\frac{k}{rt},
\vspace{-2mm}
\end{equation}
where $\frac{b}{k}$ is the base layer symbol size, $y$ is the hash  size (e.g., 32 bytes), $y(q-1)$ is the size of the partial data symbol from an intermediate layer for Merkle proof, $yq(1-r)$ is the average size of probabilistically sampled parity symbol from an intermediate layer, and $\log_{qr}\frac{k}{rt}$ is the number of layers. See Fig.~\ref{fig:Merkle_proof} for sampling response of a coded symbol on the based layer of the CMT in Fig.~\ref{fig:coded_Merkle_tree_interleaved}.

Finally, to counter selective disclosure conducted by the malicious block producer, a light node will make the $s$ requests separately, anonymously, with replacement, and with some delay between every two requests. This will prevent the malicious block producer from selectively deceiving any particular light node, or deceiving the set of light nodes that make requests at the beginning. Therefore, every light node will have the same chance to catch data availability attack.
\begin{figure*}[t]
    \centering
    \includegraphics[width=\linewidth]{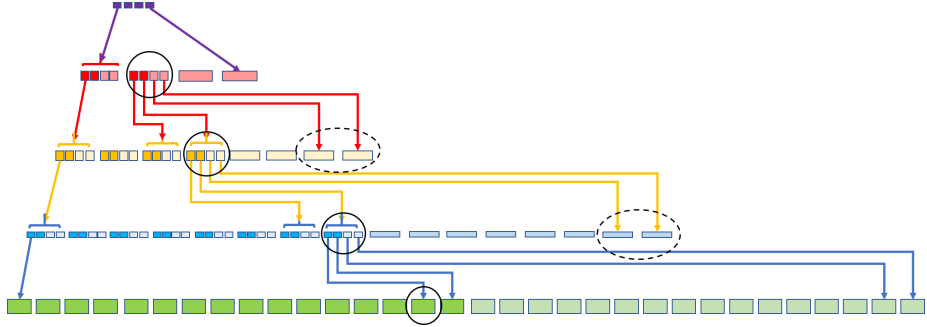}
    \caption{Merkle proof in CMT, and probabilistic sampling of parity symbols in intermediate layers. The solidly circled symbols constitute a base layer coded symbol and its Merkle proof. For the 2 dash-circled yellow parity symbols, with a probability of $1-r=0.5$, one of them (with equal chance) will be included in the proof for sampling purpose. So will the 2 dash-circled blue parity symbols.}
    \label{fig:Merkle_proof}
    \vspace{-6mm}
\end{figure*}


\vspace{-2em}

\subsection{Hash-aware peeling decoder and incorrect-coding proof}\label{sec:peeling_decoder}
A hash-aware peeling decoder is similar to conventional LDPC peeling decoder. Given the hashes of all the $n$ coded symbols and $(1-\alpha)n$ coded symbols of a layer, it iteratively solves degree-1 parity equations and check each decoded symbol against its hash and associated parity equations (Algorithm \ref{alg:peeling}). This way, the detection and proof of incorrect-coding is minimized to one parity equation.

The key condition for the peeling decoder to work is that the hashes of all the coded symbols are available. This is assured by CMT: By first downloading the root, the decoder will have all the hashes needed to decode the previous layer. Once this layer is successfully decoded, the decoded data symbols will provide all the hashes needed to decode its child layer. This top-down decoding continues until the data block is decoded, or incorrect-coding is detected at one of the layers. To prove a failed parity equation that consists of $d$ coded symbols, the decoder only needs to provide the Merkle proofs of these coded symbols, and the value of $d-1$ coded symbols. Note that the higher the failed layer, the shorter the Merkle proof of each symbol in the incorrect-coding proof.

In addition, the peeling decoder only works if 1) there are $(1-\alpha)n$ coded symbols available, and 2) that these coded symbols allow the recovery of all the $k$ data symbols. While the first condition is checked by light nodes through random sampling, the second condition requires us to find, for every layer, a erasure code whose undecodable ratio is $\alpha$ under peeling decoding. The best performance is achieved if the codes are extremely sparse (with a small $d$) and have a large $\alpha$. We now present such an ensemble of LDPC codes.


\vspace{-4mm}

\subsection{Construction of Erasure Code}
\label{sec:erasurecodeconstruction}
\vspace{-2mm}
\def\S{\mathcal S}
\def\c{\boldsymbol c}
An $(n,k)$ erasure code can be described by an $n\times (n-k)$ parity check matrix $\sH$, where each column of $\sH$ describes a parity equation, such that $\sC\sH=\overrightarrow{0}$ for any valid codeword $\sC$. In addition, every \emph{stopping set} of $\sH$ corresponds to a set of coded symbols whose hiding will prevent the full recovery of data symbols using peeling decoder.
For an $ n\times (n-k)$ parity check matrix $\sH$, a set of rows $\tau \subset [n]$ is called a {\em stopping set} if no column in $\sH_\tau$ has one non-zero element. Here $\sH_\tau$ is the submatrix of $\sH$ that only consists of the rows in $\tau$.

\begin{algorithm}[t]
\caption{Hash-aware peeling decoding algorithm}\label{alg:peeling}
\label{alg:peeling}
{\textbf{Inputs}: the hashes of all $n$ coded symbols and $(1-\alpha)n$ coded symbols;}\\
{\textbf{Initial check}: checks all the degree-0 parity equations (i.e., those whose coded symbols are all known). If any parity equation is failed, report an incorrect-coding proof and exit;}\\
\While{not all the $k$ data symbols are recovered}{
Find a degree-1 parity equation, which only has one unknown coded symbol;\\
Recover this coded symbol and verify it with its hash. If failed, report an incorrect-coding proof and exit;\\
Check all the associated degree-0 parity equations. If any parity equation is failed, report an incorrect-coding proof and exit.}
\end{algorithm}

Correspondingly, there is no parity equation that includes exactly one coded symbol among those indexed by $\tau$. Thus, if this set of coded symbols are hidden, there is no degree-1 parity equation to recover them. Since the  peeling decoder is essential for us to construct small incorrect-coding proof, the undecodable ratio $\alpha$ of a block is equivalent to the \emph{stopping ratio} of $\sH$, which is the size of the smallest stopping set divided by $n$.

\def\mJ{\boldsymbol J}

While CMT admits any erasure codes,  SPAR uses the method introduced in \cite{burshtein2004asymptotic,luby2001efficient} and analyzed in \cite{orlitsky2005stopping} for its proved ability to create, with high probability, parity matrices with a large stopping ratio. Given two integers $c$ and $d$ that satisfies $nc=(n-k)d$, we first generate an $nc\times(n-k)d$ permutation matrix $\mJ$ (a random row permutation of an $nc\times(n-k)d$ identity matrix). We then partition $\mJ$ into $n\times (n-k)$ slices, where each slice is a $c\times d$ sub-matrix. Then $\sH_{i,j}=1$ if and only if slice-$(i,j)$ contains an odd number of 1s, for $i\in[1:n]$ and $j\in[1:n-k]$. Such a random $\sH$ has the following three critical properties:
\begin{enumerate}
    \item It has a maximum row weight of $c$, and a maximum column weight of $d$;
    \item It has a non-zero probability to have a stopping ratio of at least $\alpha^*$, where $\alpha^*$ is a critical stopping ratio inherent to this method and is independent of $k$;
    \item It is NP-hard to find the minimum stopping set and determine the stopping ratio of $\sH$.
\end{enumerate}

Property 1 implies that the corresponding LDPC code has a maximum parity equation size of $d$. Property 2 implies that we could provide the same undecodable ratio (thus same sampling requirements) for all layers. Both are desirable.

Nevertheless, Property 2 and 3 together imply that we, as the developers, are not able to determine whether the LDPC codes we generate are good ($\alpha\geqslant \alpha^*$) or not, for any reasonably large $k$ (e.g., $k=1024$).

Fortunately, this problem can be easily solved through a \textbf{\em bad-code proof}. If an honest full node cannot fully decode the $k$ data symbols after receiving $(1-\alpha^*)n$ coded symbols, then this code is bad, and its small undecodable set has been found and hidden by a (very strong) malicious block producer. In this case, the honest full node can prove this code bad by broadcasting the indices of the $\alpha^*n$ coded symbols it is missing. Upon receiving and verify this bad-code proof, all the nodes in the system reject the associated block, and regenerate a code for the failed layer using an agreed random seed. This seed can be drawn from a pre-defined number sequence or the block header of a previous block, so that no consensus protocol is needed. Alternatively, distributed random number generation algorithms such as \cite{syta2017scalable} and \cite{bonneau2015bitcoin} can also be used for strictly unbiased randomness and consensus.

In other words, we solve the NP-hard problem of finding good codes by exploiting the computational resources of malicious party. Once it finds a small undecodable set and hides it, the system can easily detect this, reject the block, and update the code. This way, the system will settle at a good code for every layer eventually. As we will show in the next section, the probability of generating good codes is extremely high, so that SPAR can settle at good codes very quickly without having light nodes accept any fraud blocks. In addition, since bad code is a rare event, a heavy incentive/punishment scheme can be applied to deter false bad-code proof. Thus, the download and verification cost of bad-code proof is amortized to negligible throughout the course of the system.


\vspace{-4mm}
\subsection{Summary of the Actions of Different Node Types}
\begin{itemize}
\item \noindent {\bf Block producer (full node):} (a) It generates CMT and broadcasts the CMT root to all nodes, as well as broadcasts the entire original block (not CMT, as it can be retrieved using the original) to the full nodes only.  (b) On receiving sample requests from the light nodes, respond to them.

\item \noindent{\bf Light node:} (a) On receiving a new CMT root (or a CMT root of a pending block from a new full node), it makes separate, anonymous, and intermittent sampling requests with replacement to full nodes who claim that the block is available, as described in Section \ref{sec:samplingalgorithm}. (b) On receiving a sample, it broadcasts it to all connected full nodes. (c) If a node receives all requested samples, it assumes the block is available. (d) If a node does not receive all requested samples within a fixed time, it ``pends'' the block (i.e., keeps it in pending status). (e) If a node receives an incorrect-coding proof or bad-code proof, it rejects the block. In case of bad-code proof, it will also update the erasure code of the failed layer.


\item \noindent{\bf Other full node.} (a) On receiving valid samples, it tries to recover the data block through both downloading the original data block from the other full nodes and collecting coded symbols forwarded by the light nodes. It will decode the tree from top to bottom using a hash-aware peeling decoder. (b) It rebroadcasts the received valid samples. (c) If an incorrect coding or a bad code has been detected, it will send the corresponding proof and reject this block. (d) If it has received/fully decoded a data block and verified it, it will declare the availability of this block to all other nodes and respond to sample requests from light nodes. 
\end{itemize}

\vspace{-4mm}
\section{Performance Analysis}\label{sec:performance}
\label{sec:performance}
\vspace{-2mm}
\subsection{Security}
\begin{theorem}\label{thm:security}
In SPAR,  a block producer cannot cause the soundness and agreement to fail with a probability lower than
\begin{align*}
   P_f \leq \max \left\{(1-\alpha_{\min})^s, \quad 2^{\underset{i}{\max}\left[H(\alpha_i) n_i-ms\log\frac{1}{1-\alpha_i}\right]} \right\}.
\end{align*}
Here $n_i$ and $\alpha_i$ are the number of symbols and undecodable ratio on the $i$th layer of CMT, $\alpha_{\min} \triangleq \min_i \alpha_i$, and $s$ is the number of coded symbols each light node samples from the base layer.
\end{theorem}

\begin{proof}
{\em Soundness}:  Soundness fails if a light node thinks that a block is available, and no full node is able to reconstruct the entire coded Merkle tree. We note that the reconstruction fails if any layer of the CMT cannot be recovered correctly. Let us focus on a single layer $i$ with $n_i$ coded symbols and an undecodable ratio of $\alpha_i$, and assume that the malicious block producer hides $\alpha$ fraction of the coded symbols (and does not respond to requests for those symbols).

Case-1: Consider the case of $\alpha \geq  \alpha_i$. The probability of soundness failure for a node is given by the probability that a node receives all $s$ symbols that it samples, this probability is $(1-\alpha)^s \leq (1-\alpha_i)^s \leq (1-\alpha_{\min})^s$.  

Case-2: Consider the case of $\alpha < \alpha_{i}$. The soundness failure occurs if a full node cannot decode the entire block or is unable to furnish a incorrect-coding proof. The full node will fail to accomplish these tasks only when it is able to receive fewer than  $1-\alpha_{i}$ fraction of symbols. Define $Z_i$ to be the total number of distinct symbols collected by the honest full node ($Z_i \in \{0,1,..,n_i\}$). Let $m$ be the total number of light nodes, then $m \cdot s$ is the total number of i.i.d. samples. Now we have 
\begin{align}
    P(Z_i \leq (1-\alpha_i)n_i) &\leq {n_i \choose \alpha_i n_i}\frac{(n_i-\alpha_in_i)^{ms}}{n^{ms}},\label{eq:count}\\
    &\leq 2^{H(\alpha_i) n_i}(1-\alpha_i)^{ms}, \\
    &= 2^{H(\alpha_i) n_i-ms\log\frac{1}{1-\alpha_i}} \leq 
    2^{\underset{i}{\max}\left[H(\alpha_i) n_i-ms\log\frac{1}{1-\alpha_i}\right]}.
\end{align}
Here (\ref{eq:count}) is by counting the number of sampling instances that provide less then $(1-\alpha_i)n_i$ distinct symbols. $H(p) = p\log\frac{1}{p} + (1-p)\log\frac{1}{1-p}$ is the binary entropy function. It is apparent that we would need $m$ large to make the above bound vanish exponentially with $s$.

The probability of soundness failure is smaller than the maximum probability of the two cases. 

{\em Agreement}: We will argue here that soundness implies agreement for our protocol. As defined, soundness ensures that a honest full node is able to decode the block. Once a honest full node decodes the block, it will let all light nodes know that it has that block. The light nodes have either already accepted the block or have ``pend''-ed the block (the light nodes could not have rejected the block since it is a valid block). The light nodes that pended the block will query with the honest full node and eventually accept the block. Thus soundness implies agreement (since now every light node agrees on the availability of the block). 
\end{proof}

Theorem~\ref{thm:security}
implies that the security of SPAR increases exponentially with the number of samples each light node takes ($s$), when the number of light nodes ($m$) is linear with the block size. 

\vspace{-5mm}
\subsection{Costs and Complexity}
\vspace{-2mm}
A light node has three download costs: 1) the header, which is the CMT root of size $t$; 2) the random sampling, and 3) the incorrect-coding proof. In CMT, the header is the CMT root of size $t$. The sampling cost can be computed using the average parity-symbol-sampled Merkle proof size given in (\ref{eq:merkle_proof_with_parity}) to be:
\begin{eqnarray}
s \left(\frac{b}{k} + y\left(2q-1-qr\right) \log_{qr}\frac{k}{rt}\right)=O(\log k)=O(\log b),
\label{eqn:sampling_cost}
\end{eqnarray}
where $b$ is the size of a block, and the equations hold due to that 1) $s$ is a constant; and 2) $b/k$ is the base layer symbol size, which is a constant. The incorrect-coding proof size can be similarly computed as

\vspace{-0.5em}

\begin{eqnarray}
\frac{(d-1)b}{k} + dy(q-1) \log_{qr}\frac{k}{rt}=O(\log b),
\label{eqn:incorrect_proof_cost}
\vspace{-0.5em}
\end{eqnarray}
where the first term is the size of $d-1$ coded symbols, and the second term is the size of $d$ Merkle proofs.
Finally, since the hash-aware peeling decoder decodes one coded symbol using $d-1$ coded symbols in one parity equation, the decoding complexity is $O(1)$ per symbol and, thus, is $O(b)$ in total.

\subsection{Choice of parameters}

Our first key parameter is the coding rate $r$. A smaller $r$ means more parity symbols and thus a potentially larger undecodable ratio and less sampling. But it will also increase the height of CMT, the Merkle proof size, and the decoding complexity. For a reasonable tradeoff, we choose $r = 0.25$.

Given $r$, the next two parameters we should decide are a pair $(c,d)$ that satisfies $c/d = 1- r= 0.75$ for the random LDPC code generator, where $d$ is the maximum parity equation size. This gives us $c = 0.75d$ and requires us to find the critical undecodable ratio of the ensemble as a function of $d$, which is provided in Table \ref{tab:stopping_versus_d} based on the analysis in \cite{orlitsky2005stopping}.

\vspace{-1em}
\begin{table}[htbp]
\vspace{-6mm}
    \centering
    \caption{The critical undecodable ratio $\alpha^*$ as a function of $d$}
    \begin{tabular}{|c|c|c|c|c|c|}
    \hline
    $d$     & 4 & 8 & 12 & 16 & 20\\
    \hline
    $\alpha^*$ &   0.0795&   0.124 &  0.111&
   0.0981 &
   0.0877\\
    \hline
    \end{tabular}
    \label{tab:stopping_versus_d}
\vspace{-6mm}
\end{table}

Evidently, $d=8$ maximizes the critical undecodable ratio. In addition, it also admits a small incorrect coding proof that only requires only 7 coded symbols and 8 Merkle proofs. As a result, we choose $(c,d) = (6,8)$.

\subsection{How quickly does SPAR settle at a good erasure code?}
Due to random code generation, each layer of SPAR  eventually settles at a good code (with an undecodable ratio of at least $\alpha^*$) after a few bad codes have been deployed by the system and potentially utilized by malicious block producer to hide the data. We study the number of such attacks (note that they will never succeed) and updates before SPAR can settle: intuitively, this number can be computed as  $(1-P(\alpha<\alpha^*))^{-1}-1$, where $P(\alpha<\alpha^*)$ is the probability that a randomly generated code has an undecodable ratio smaller than $\alpha^*=0.124$. Using an upper bound on $P(\alpha<\alpha^*)$  characterized in~\cite{orlitsky2005stopping} we can derive the settlement speed of SPAR as below. We note that most of the layers of CMT will immediately settle at a good code upon  launching. The only exception is the layer with $n=256$, which will settle after 7.7 bad codes, but without any fraudulent blocks been accepted. The proof is in Appendix~\ref{sec:failure_probability}. 


\begin{theorem}
\label{thm:probs}
Using the random $(6, 8)$-LDPC code in Section \ref{sec:erasurecodeconstruction}, the expected number of bad erasure codes ($\alpha<\alpha^*=0.124$) a CMT layer with $n$ coded symbols will use before it settles at a good code ($\alpha\geq 0.124$) is approximated in Table \ref{tab:expected_successful_attack}. 

\begin{table}[htbp]
\vspace{-8mm}
\centering
\caption{Number of bad code before settlement.}
\begin{tabular}{|c|c|c|c|c|c|c|}
\hline 
    $n$ & 256 & 512 & 1024 & 2048 & 4096 & $>4096$\\\hline
    $P(\alpha<\alpha^*)$ & 0.886 & 5.3e-2 & 2.0e-3 & 1.3e-3 & 3.2e-4 & $<$3.2e-4\\\hline
    \# bad codes & 7.7 & 0.06& 0.002 &0.001 & 0.0003 & $<0.0003$\\
    \hline
\end{tabular}
\label{tab:expected_successful_attack}
\vspace{-8mm}
\end{table}
\end{theorem}


\vspace{-4mm}
\section{Implementation for Bitcoin and Experiments}\label{sec:numerical}
\vspace{-2mm}
We developed in {\sf Rust} a {\sf Coded Merkle Tree} library~\cite{cmt-lib} for {\sf Parity Bitcoin}~\cite{parity-btc} clients (see Appendix \ref{sec:cmt_library} for more details). Our library integrates seamlessly into the current mainnet implementation, and requires minimal change on the block data structure (only need to add CMT root to the block header). Note, however, that this change is incompatible with existing Bitcoin clients. Developing a SPAR-protected Bitcoin testnet and a Bitcoin Improvements proposal (for Bitcoin Core) are ongoing research activities outside the scope of this paper. 


We combine the CMT library with the performance analysis in Section \ref{sec:performance}, and numerically evaluate SPAR's light node download costs (header, sampling, and incorrect-coding proof) and full node decoding speed, for a wide range of block sizes (Table \ref{tab:numierical_para}), and compare them with the 2D-RS based solution proposed in~\cite{mustafa_fraud_proof} using its {\sf C++}/{\sf Python} implementation \cite{dataAvailable,rs_vitalik_repo}.

\vspace{-2em}
\begin{table}[h]
\caption{Experiment Parameter Configuration}
    \centering
    \begin{tabular}{|c|c|c|}
        \hline
         \textbf{parameter} & \textbf{value} & \textbf{notes}\\\hline
         symbol size (B) & 256 & ~\\\hline
         base layer $k$ & $2^{12}$ to $2^{22}$ & block size is thus 1 to 1024 MB\\\hline
         coding rate $r$ & 0.25 & thus $n=4k$\\\hline
         hash size (B) & 32 & SHA256 is used\\\hline
         target confidence & 99\% & each light node keeps sampling until it is 99\% confident\\\hline
         \multicolumn{3}{|c|}{\textbf{SPAR specific parameters}}\\\hline
         LDPC sparsity & (6, 8) & each parity equation has at most 8 coded symbols\\\hline
         stopping ratio $\beta$ & 0.124 & $0.124n$ symbols must be hidden to prevent decoding\\\hline
         batching factor $q$ & 8 & CMT layer size reduction rate is $qr=2$ as ordinary trees\\\hline
         CMT root size $t$ & 256 hashes & the same as 2D-RS header size for 1MB blocks\\\hline
    \end{tabular}
    \label{tab:numierical_para}
\end{table}

\vspace{-2em}
\textbf{Header} (Fig. \ref{fig:header_cost}): A SPAR light node only downloads fixed $t=256$ hashes in header, whilst 2D-RS requires $1+2\sqrt{n}$. Thus, the header download cost of SPAR becomes much smaller than 2D-RS with growing block size. For a 64MB block, the cost is only 0.01\% of the block size in SPAR, but is 0.1\% in 2D-RS.

\textbf{Incorrect-coding proof} (Fig. \ref{fig:proof_cost}): A SPAR incorrect-coding proof only involves $d-1=7$ coded symbols and their Merkle proofs, whilst 2D-RS requires $\sqrt{k}$. Thus, the incorrect-coding proof download cost of SPAR becomes much smaller than 2D-RS with growing block size. For a 64MB block, the cost is only 0.051\% of the block size in SPAR, but is 0.48\% in 2D-RS.

\textbf{Sampling cost} (Fig. \ref{fig:sampling_cost}):
2D-RS has a higher undecodable ratio of $25\%$ compared to SPAR's $12.4\%$. Thus, for 99\% confidence, $s=17$ distinctive samples are enough in 2D-RS, whilst SPAR requires $s=35$ if the adversary is strong enough to find, with NP-hardness, the size-$0.124n$ stopping set. But under a realistically weak adversary that randomly selects CMT symbols to hide, SPAR only requires $s=8$ because our LDPC ensemble can tolerate an average of $47\%$ missing symbols. On the other hand, the over-sampling of each layer increases SPAR's sampling cost. Thus, although both techniques offer $O(\log k)$ sampling costs that quickly reduces with growing block size, the cost of SPAR is about 10$\sim$16 (resp. 2.5$\sim$4) times of 2D-RS under strong (resp. weak) adversaries. However, in practice, one can further reduce SPAR sampling cost by increasing the header size $t$, thus reducing the size of the Merkle proof of each symbol. 

\textbf{Decoding speed} (Fig. \ref{fig:decoding_speed}): SPAR's sparse and binary encoding, at its current implementation level is already   over 10 times faster than 2D-RS for all the tested block sizes. 




\begin{figure}[t]
\centering
\subfigure[Header cost]{\includegraphics[width=0.49\linewidth]{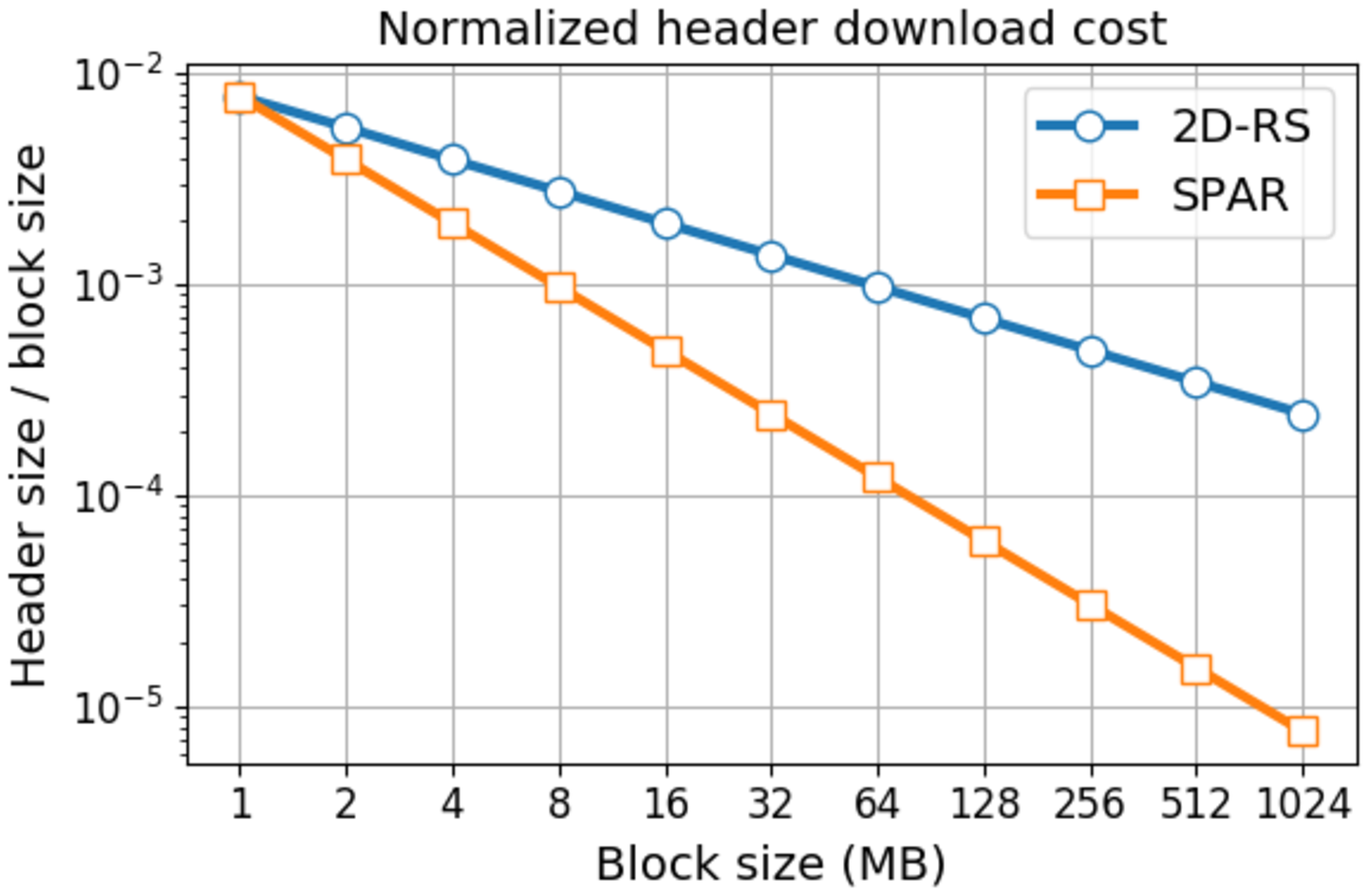}\label{fig:header_cost}}
\subfigure[Incorrect-coding proof size]{\includegraphics[width=0.49\linewidth]{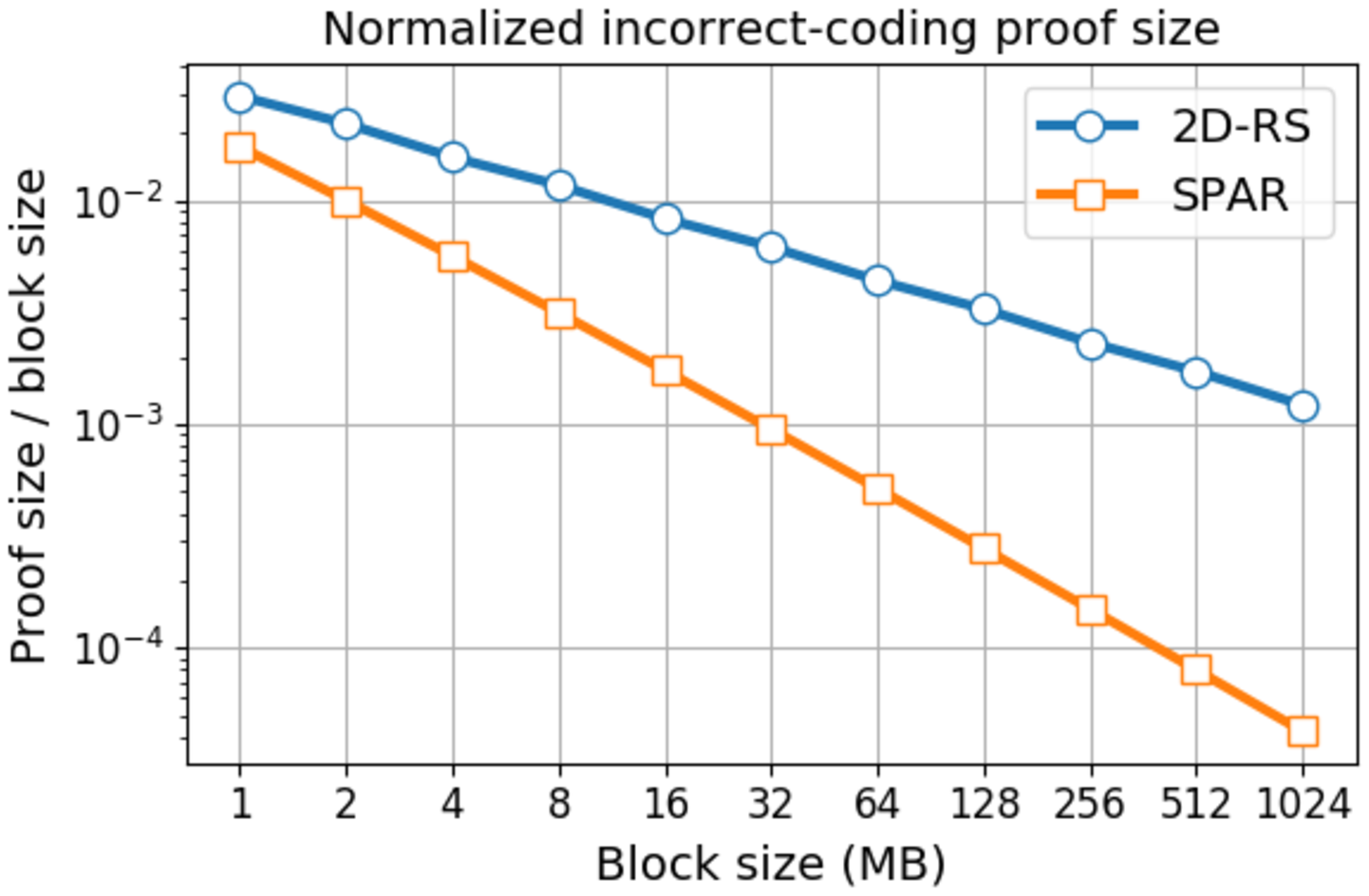}\label{fig:proof_cost}}
\subfigure[Sampling cost]{\includegraphics[width=0.49\linewidth]{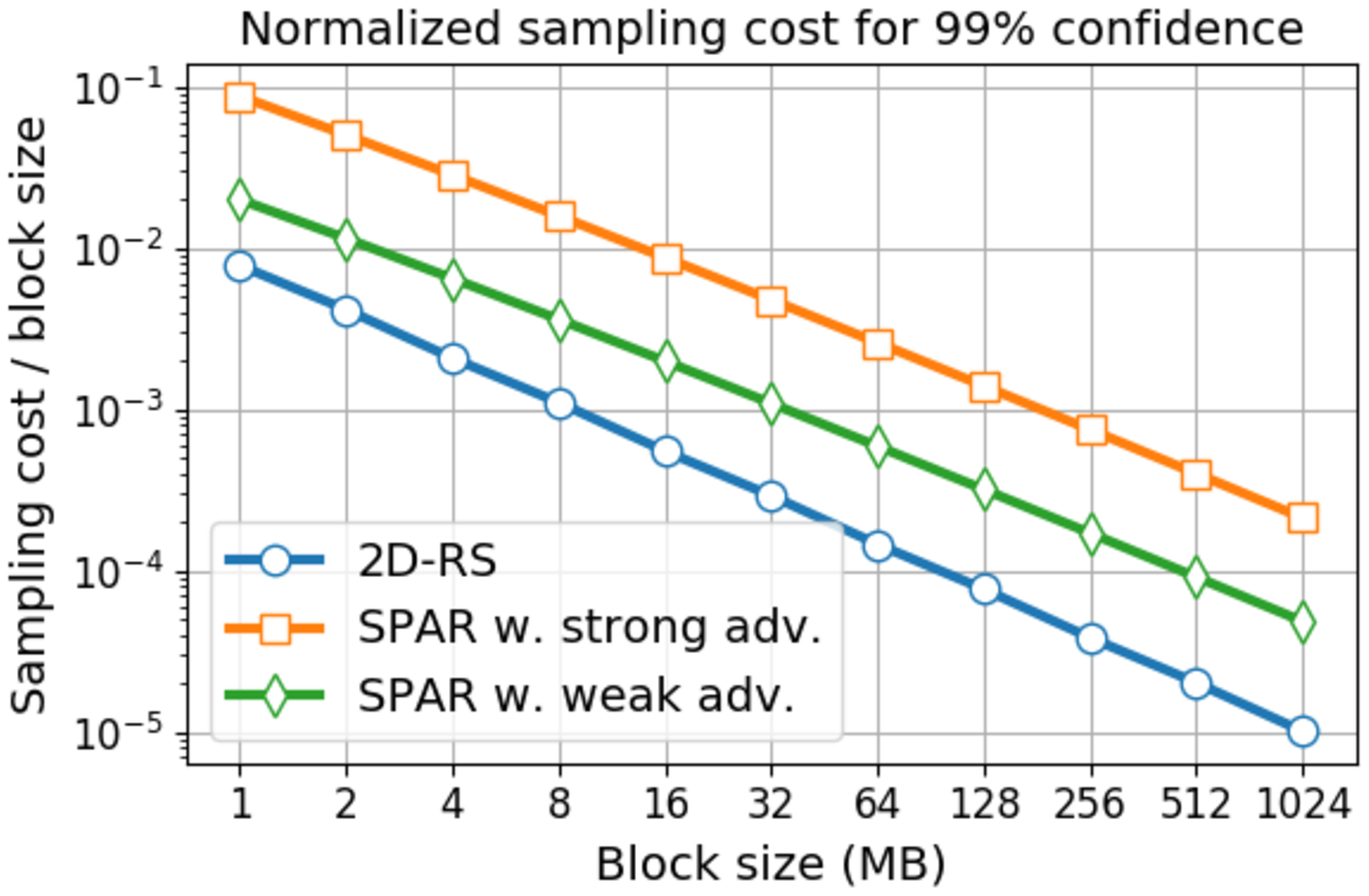}\label{fig:sampling_cost}}
\subfigure[Decoding speed]{\includegraphics[width=0.49\linewidth]{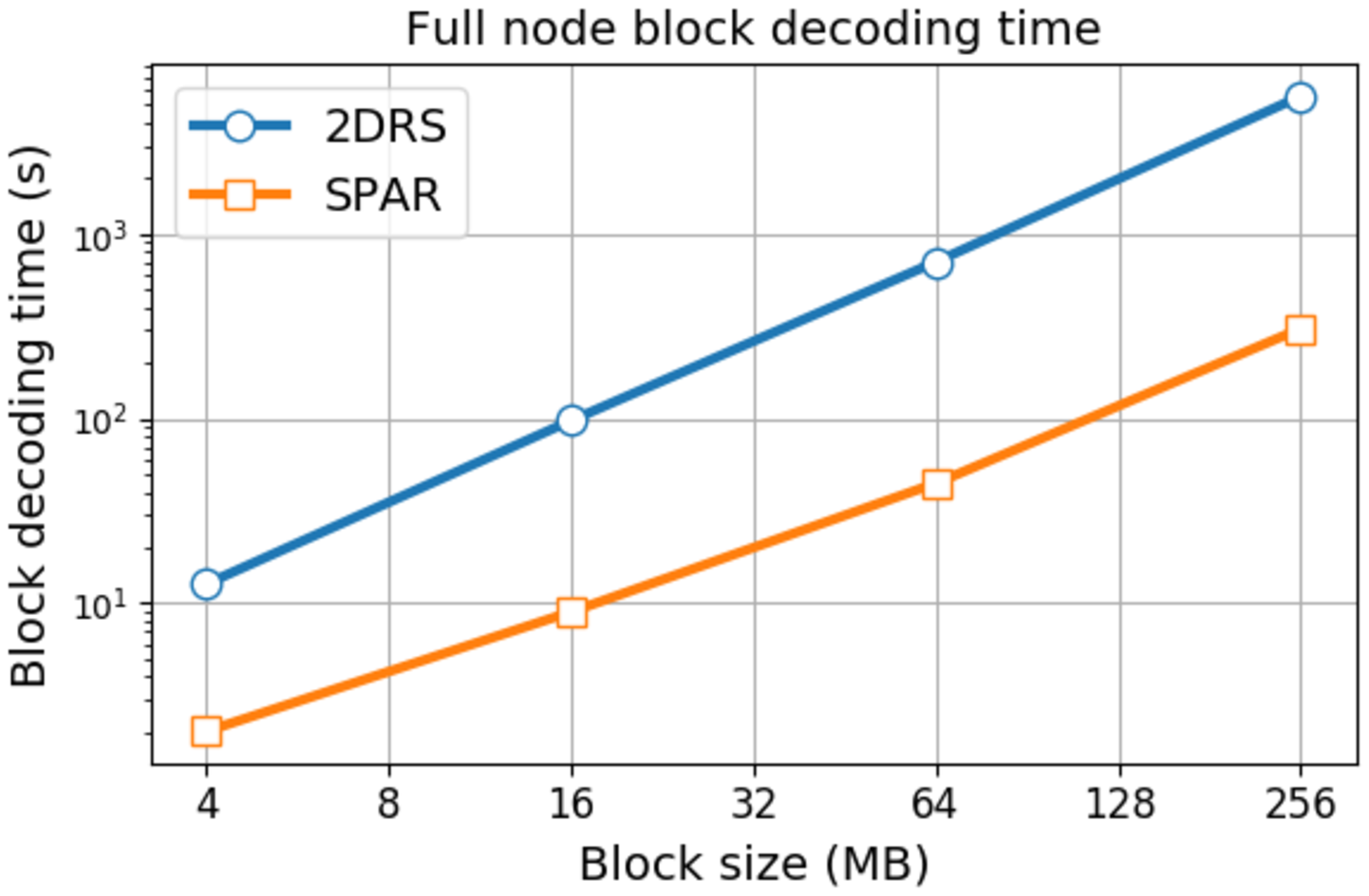}\label{fig:decoding_speed}}
\caption{Communication costs and decoding speed of SPAR and 2D-RS.}
\vspace{-6mm}
\end{figure}
\vspace{-6mm}
\section{Conclusion and Discussions}
\vspace{-4mm}
By iteratively applying a special ensemble of LDPC codes to every layer of a Merkle tree and batching the hashes of each coded layer into the data symbols of the next layer, we invented a novel hash accumulator called coded Merkle tree (CMT). Built upon CMT, we proposed a novel data availability verification system called SPAR,  which allows the availability and integrity of the entire tree to be checked at constant costs. 


SPAR can play a key role in scaling blockchain systems that incorporate light nodes because it empowers these nodes with real-time verification of data availability and integrity at small and constant costs. SPAR can also be used to scale the communication of sharded blockchain systems (e.g.,~\cite{luu2016secure,kokoris2018omniledger,li2018polyshard}), where full nodes of one shard operate as light nodes of other shards, as SPAR allows them to efficiently check the availability and integrity of blocks in other shards. 

Integrating SPAR into existing blockchain systems requires minimum changes and no extra bandwidth consumption. An honest block producer only needs to broadcast the original data block as usual and attach the CMT root in the block header. This is sufficient for other full nodes to reproduce the CMT and offer sampling services for light nodes. Our library for CMT  in {\sf Rust} for {\sf Parity Bitcoin} clients maintains the same API  as the standard Merkle tree module. Noting that classic Merkle trees are indeed special CMTs with coding rate $r=1$ and batching factor $q=2$, our library readily replaces the standard module and is backward compatible. 

\section*{Acknowledgement}
This research was conducted under the auspices of Trifecta Blockchain Inc.

\bibliographystyle{splncs04}
\bibliography{CMT_arxiv.bib}

\appendix
\section{Proof of  Theorem~\ref{thm:probs}}\label{sec:failure_probability}

Based on the proof of Theorem 8 of \cite{orlitsky2005stopping}, we know that for an $(n,k)$ LDPC code that is randomly chosen from a $(c,d)$ ensemble, described in Section \ref{sec:erasurecodeconstruction}, 
the probability that the stopping distance of the code is smaller than $\alpha^* n$ is upper-bounded by 

\begin{eqnarray}\label{eq:stop_distance}
P(\alpha n < \alpha^* n) \le \min\left\{\inf\limits_{0 < \delta < \alpha^*} \left(\! n(\alpha^* - \delta)e^{n\max \limits_{\theta\in[\delta,\alpha^*]}\gamma(\theta)} \!\!+\!\! \sum_{i = 1}^{\delta n -1} I_i \!\right),1\right\}
\end{eqnarray} 
where for $c = d(1-r)$, $h(\theta) = -\theta \log \theta -(1-\theta)\log(1-\theta)$, and $x_0$ as the only positive solution to $\frac{x(1+x)^{d-1}-x}{(1+x)^d  -dx} = \theta$,
\begin{align*}
    \gamma(\theta) & \triangleq \frac{c}{d}\log \left( \frac{(1 + x_0)^d - dx_0}{x_0^{\theta d}}\right) - (c-1)h(\theta),\\
    I_i & \triangleq {n\choose i}\frac{{n(\frac{c}{d} + \frac{\delta c}{2} - \frac{\delta c}{d})\choose \lfloor\frac{ic}{2}\rfloor }(2d-3)^{ic}}{{nc \choose ic}}.
\end{align*}

For $\alpha^* =0.124$, the above upper bound for small $n$ (e.g., $n=256$) becomes degenerated (i.e., reduces to the trivial bound of $1$). In order to obtain a good approximation for all the values of $n$, we approximate the upper bounds on $P(\alpha < 0.124)$ using a slightly smaller undecodable ratio of $0.116$. Then, we evaluate the upper bounds on $P(\alpha n < 0.116 n)$ in (\ref{eq:stop_distance}), for all the considered values of $n$ in Theorem~\ref{thm:probs} to obtain the probabilities in the second row of Table~\ref{tab:expected_successful_attack}.

We note that since 0.116 is very close to 0.124, SPAR's inherent oversampling of intermediate layers will provide sufficient protection for data availability on these layers, so that the light node sampling cost will not increase. 




\section{{\sf Coded Merkle Tree} Library}\label{sec:cmt_library}
We developed in {\sf Rust} a {\sf Coded Merkle Tree} library~\cite{cmt-lib} for {\sf Parity Bitcoin}~\cite{parity-btc} clients. We modify the data structure of the block header to add a new field {\sf coded\_merkle\_roots\_hashes}, which are the hashes of the coded symbols on the last level of the coded Merkle tree constructed from this block.

To use the {\sf Coded Merkle Tree} library on a block, we require the following input parameters from the users:
\begin{itemize}
    \item {\sf BASE\_SYMBOL\_SIZE}: size of a symbol on the base level, and hence the number of systematic symbols on the base level.
    \item {\sf AGGREGATE}: number of hashes to aggregate into a symbol on the next level.
    \item {\sf HEADER\_SIZE}: number of hashes stored in the block header. This also decides the total number of levels in the coded Merkle tree.
    \item Codes for all levels of coded Merkle tree, in forms of sparse representations of their parity-check matrices. 
\end{itemize}

Given the above parameters, {\sf Coded Merkle Tree} implements the following key functionalities:
\begin{itemize}
    \item {\sf coded\_merkle\_roots}: construction of the coded Merkle tree from the block content.
    \item {\sf merkle\_proof}: generating the Merkle proof for any symbol on any level of coded Merkle tree. By design, this returns a set of symbols on the higher level.
    \item {\sf sampling\_to\_decode}: sampling symbols on the base level, together with their Merkle proofs.
    \item {\sf run\_tree\_decoder}: decode the entire coded Merkle tree level by level from the roots. Each level is decoded by running a hash-aware peeling decoder, using the decoded symbols on the previous level as the hash commitments.
    \item {\sf generate\_incorrect\_coding\_proof}: 1) when a coding error is detected, this function returns $d-1$ symbols in a parity equation, and Merkle proofs for all $d$ symbols in that equation; 2) when the peeling process gets stuck before all symbols are decoded, this function returns the indices of the missing symbols as a stopping set.
\end{itemize}

\end{document}